\documentclass[11pt]{article}%
\usepackage{amsfonts}
\usepackage{amssymb}
\usepackage{graphicx}
\usepackage{setspace}
\usepackage[round]{natbib}
\usepackage{amsmath}%
\usepackage{amsthm}%
\usepackage{bbm}
\usepackage{palatino}
\usepackage{paralist}
\usepackage{multirow}
\usepackage{booktabs}
\usepackage{dsfont}
\usepackage{indentfirst}
\usepackage{enumerate}
\usepackage{enumitem}
\usepackage{longtable}
\usepackage{bm}
\usepackage{caption}
\usepackage{subcaption}
\usepackage{mwe}
\usepackage{float}
\usepackage{booktabs}
\usepackage{lscape}
\usepackage[hyperindex,breaklinks]{hyperref}
\hypersetup{colorlinks=true,       % false: boxed links; true: colored links
	linkcolor=red,       
	citecolor=blue,        % color of links to bibliography
	filecolor=magenta,      % color of file links
	urlcolor=cyan           % color of external links
}  
\newcommand{\var}{\mathrm{var}}
\newcommand{\cov}{\mathrm{cov}}

\newtheorem{theorem}{Theorem}

\theoremstyle{definition}

\DeclareMathOperator*{\maxf}{maximize \quad}
\newcommand{\maxdisp}[3]{\begin{align*}\maxf_{#1} & #2\\\st  & #3\end{align*}}

\title{Deep Learning Partial Least Squares}

\author{	
	\makebox[.4\linewidth]{Nicholas Polson\footnote{Email: ngp@chicagobooth.edu}}\\
	\textit{\small  Booth School of Business}\\
	\textit{\small  University of Chicago}\\
	\and 
	\makebox[.4\linewidth]{Vadim Sokolov}\\
	\textit{\small  Department of Systems Engineering }\\
	\textit{\small  and Operations Research}\\
	\textit{\small  George Mason University}\\
	\and
	\makebox[.4\linewidth]{Jianeng Xu}\\
	\textit{\small  Booth School of Business}\\
	\textit{\small  University of Chicago}\\
}

\graphicspath{{../fig/}{../wine/}{fig/}}

\begin{document}
\maketitle

%%%%%%%%%%%%%%%%%%%%%
\begin{abstract}
%%%%%%%%%%%%%%%%%%%%%

\noindent 
High dimensional data reduction techniques are provided by using partial least squares within deep learning. Our framework provides a nonlinear extension of PLS together with a disciplined approach to feature selection and architecture design in deep learning. This leads to a statistical interpretation of deep learning that is tailor  made for predictive problems. We can use the tools of PLS, such as scree-plot, bi-plot to provide model diagnostics. Posterior predictive uncertainty is available using MCMC methods at the last layer. Thus we achieve the best of both worlds: scalability and fast predictive rule construction together with uncertainty quantification. Our key construct is to employ deep learning within PLS by predicting  the output scores as a deep learner of the input scores. As with PLS our X-scores are constructed using SVD and applied to both regression and classification problems and are fast and scalable. Following \cite{frank_statistical_1993}, we provide a Bayesian shrinkage interpretation of our nonlinear predictor. We introduce a variety of new partial least squares models: PLS-ReLU, PLS-Autoencoder, PLS-Trees and PLS-GP. To illustrate our methodology, we use simulated examples and the analysis of preferences of orange juice and predicting wine quality as a function of input characteristics. We also illustrate Brillinger's estimation procedure to provide the feature selection and data dimension reduction.  Finally, we conclude with directions for future research. 


\end{abstract}

\newpage

\section{Introduction}

% Deep learning (DL-PLS) is viewed through the lens of partial least squares.

% Brillinger and Stein's Lemma
Many high dimensional regression problems  involve identifying a functional that maps a high-dimensional input variable $x \in R^{p}$ to an output variable $y \in R^q$,
\[
y = \phi(x) + \epsilon.
\]
The function $\phi$ is estimated using the observed matrices of inputs $X\in R^{n \times p}$ and corresponding outputs $Y \in R^{n \times q}$. A high-dimensional input $x$ naturally arise in many applications, such as marketing, genetics and robotics. For example, in business applications, it is typical to increase the dimensionality of the input $x$  by introducing new hand-coded predictors \citep{Hoadley2000}.  This expanded input can include terms such as interactions, dummy variables or nonlinear functionals of the original predictors.  

Many high dimensional statistical models require data dimension reduction methods. Following \cite{breiman2001statistical} we represent data as being generated by a black box with an input vector $x$. The black box generates output $y$, or a predictive distribution over the output $p(y \mid x)$ that describes uncertainty in predicting $y$ from $x$. \cite{fisher1922mathematical} and \cite{cook2007fisher} describe clearly the issue of dimension reduction. Although it was typical to find predictors via screening and plotting them against the output variable \cite{Li2007}, Fisher argued with this approach and said in a classic empirical analysis, ``The meteorological variables to be employed must be chosen without reference to the actual crop record". Many researchers \citep{cox68} since then showed counterexamples to Fisher's approach, when dimensionality reduction techniques that do not use the output variable lead to poor results \cite{fearn_misuse_1983}. In particular, the output variables should be used for dimension reduction when the input variables are controlled by the experimenter \citep{cox2000theory}.

Partial least squares (\cite{wold_soft_1975} and \cite{wold_collinearity_1984}) has been applied in many fields, in particular, chemo-kinetics, but has received little attention in machine learning. On the other hand, nonlinear dimensionality reduction techniques for feature selection such as deep learning are prevalent. Our goal here is to incorporate deep learning into partial least squares to provide a framework for statistical interpretation and predictive uncertainty in high dimensional problems. Thus we provide a very flexible nonlinear extension of PLS that merges the two cultures of black box and statistical modeling as proposed by \cite{breiman2001statistical}. 
In doing so, we introduce a variety of new partial least squares models: PLS-ReLU, PLS-Autoencoder, PLS-Trees and PLS-GP.

Our work builds on the literature of nonlinear PLS \citep{wold_collinearity_1984,Fra94,malthouse1997nonlinear,Eriksson2009,liu2019dynamic}. Deep Learners model $\phi$  as a superposition of univariate affine functions  \citep{polson_deep_2017} and are universal approximators. Whilst Gaussian Process \citep{gramacy2008bayesian,higdon2008computer} are also universal approximators and can capture relations of high complexity, they typically fail to work in high dimensional settings. Tree methods can be very effective in high dimensional problems. Hierarchical  models are flexible stochastic models but require high-dimensional integration and MCMC simulation. Deep learning, on the other hand, is based on scalable fast gradient learning algorithms such as Stochastic Gradient Descent (SGD) and its variants.  Modern computational techniques such as automated differentiation (AD) and accelerated linear algebra (XLA) are available to perform stochastic gradient descent (SGD) at scale \citep{polson2020deep}, thus avoids the curse of dimensionality by simply pattern matching and using interpolation to predict in other regions of the input space. User friendly libraries such as  TensorFlow or PyTorch provide high level interfaces and  make deep learning modeling accessible without requiring to implement any of the low level linear algebra routines. The algorithmic culture has achieved much success in high dimensional problems. 

Rather than using shallow additive architectures as in many statistical models, deep learning uses layers of semi affine input transformations to provide a predictive rule. Applying these layers of transformations leads to a set of attributes (a.k.a features) to which predictive statistical methods can be applied. Thus we achieve the best of both worlds: scalability and fast predictive rule construction together with uncertainty quantification.  
Sparse regularisation with un-supervised or supervised learning finds the features. We clarify the duality between shallow and wide models such as PCA, PPR and deep but skinny  architectures used in deep learning such as auto-encoders, MLPs, CNN, and LSTM. The connection with data transformations is of practical importance for finding good network architectures. 

There are a number of advantages of our methodology over traditional deep learning. First, PLS uses linear methods to extract features and weight matrices. Due to a fundamental estimation result of \cite{brillinger_generalized_2012}, we show that these estimates can provide a stacking block for our DL-PLS model. Moreover, traditional diagnostics (bi-plots, scree-plots) to diagnose the depth and nature of the activation functions allow us to find a parsimonious model architecture. On the theoretical side, we follow \cite{frank_statistical_1993} and provide a Bayesian shrinkage interpretation of our model. This allows us to discriminate between unsupervised and supervised learning method.

On the empirical side, we illustrate our methodology  using simulated examples and analysis of preferences of orange juice data and predicting wine quality as a function of input variables. We also illustrate the novel estimation procedure suggested by Brillinger to provide the feature selection and data dimension reduction. 
\cite{lindley1968} provides an alternative fully Bayesian decision theoretic approach to select the number of components in the model. 
\cite{hahn_partial_2013} provide a full Bayes analysis of stochastic factor models. 

The rest of the paper is outlined as follows.  Section 1.1 provides definitions, notation and connections with previous work.
Section 2 describes partial least squares from the perspective of Bayesian shrinkage. Section 3 provides our general modeling strategy of deep learning within partial least squares. 
We provide a nonlinear extension of partial least squares by using deep layer to model the relationship between the scores generated in PLS. 
Section 4 provides applications,  We illustrate a key theoretical result of Brillinger on how to estimate deep layered models and show how to recover ReLU and 
tanh activation functions.  The orange juice dataset is re-analysed and we find evidence for nonlinearity in the scores architecture. 
Section 5 concludes with directions for future research.

\subsection{Connection with Previous Work} 

% Supervised statistical learning given data as a sequence of input-output pairs. The goal is to determine a predictive rule for new input cases or to describe with statistical uncertainty is described by  predictive distribution of future observations. How does one achieve such a goal given the empirical data of $N$ input-output pairs $(Y_i,X_i)_{i=1}^N$. We need to formulate an input-output map, $Y=G(X)$, the question is how to construct $G$?

% The machine learning problem then is to find a good predictor rule from a training dataset if input-output pairs $(Y_i,X_i)_{i=1}^N$ of observed data. 
% The goal is to predict at a new high dimensional $ X_i =  ( X_{1i} , \ldots , X_{pi} ) $.
% To achieve good generalisability we need to be able to perform nonlinear dimension reduction and to find a suitable set of features/factors. Deep learners together with sparse optimisation provides such a framework. 
From a probabilistic view point, it is natural to view input-output pairs as being generated from some joint distribution
\[
(y_i,x_i) \sim p(y,x), \; i=1,\dots,N.
\]
The joint distribution $ p( y , x) $ is characterised by its two conditionals distributions 
\begin{itemize}
\item  $ p( y \mid x)  $, which can be used for prediction $ \hat{y} = E(y\mid x) $ together with full predictive uncertainty quantification using the 
probabilistic structure.
\item  $ p(  x \mid y ) $.  This conditional is  high-dimensional and we need to perform dimension reduction and find an efficient set of nonlinear features to be used as predictors in step 1. Selection using sparsity and deep learners are typical methods used.
\end{itemize}

In our PLS-DL model we use a hierarchical linear regression model as the last layer. This allows us to perform a full Bayes inference at the level, namely calculating the predictive $ p( y \mid x ) $  helps with assessing prediction uncertainty. Using stochastic regressors allows us to use inference methods proposed by \cite{brillinger_generalized_2012}.

There are several types of \textit{projection-based dimensionality deduction techniques in regression}. All of those rely on finding orthogonal projections of columns of $X$ and (possibly) $Y$ into a lower dimensional liner space. Essentially, the problem us to find a new basis and to find the projections into this basis. The basis vectors are called \textit{loadings} and the projection is called the \textit{score}. 

Principal Components is arguable the most widely used dimensionality reduction technique in high-dimensional settings. The dimensionality reduction is achieved by changing a basis and using singular vectors of the input matrix $X$ as new basis vectors. The singular vector decomposition $X = USV^T$, where $V$ is $p \times p$ orthogonal matrix, $U$  is $n \times n$ orthogonal and $S$ is $n \times p$ diagonal matrix, provides the new set of basis vectors (directions) $V$ and the projections into this new basis $XV = US$ are called the principal components (scores). Then dimensionality reduction is achieved by projecting inputs into space spanned by the first $L<p$ right singular vectors $V_L = (v_1,\ldots,v_L)$ which correspond to the $L$ largest singular values $s_1\le s_2,\ldots,\le s_L$. The resulting projection leads to the best $L$-rank approximation of the input matrix $X$, $X \approx U_LS_LV_L^T$ \citep{golub2013matrix}.

Principal component dimension reduction is independent of $y$ and can easily discard information that is valuable for predicting the desired output. Sliced inverse regression (SIR)~\citep{li1991sliced,cook2009regression,jiang2014variable} overcomes this drawback and assumes that conditional distribution can be indexed by $L$ linear combinations
\[
p(y\mid x) = p(y\mid Wx),
\]
and by computing the set of $L$ basis vectors (rows of $W$) using both $x$ and $y$.

\textit{Partial Least Squares} (PLS) also uses both $X$ and $Y$ to calculate the projection, further it simultaniously finds projections for both input $x$ and output $y$, it make it applicable to the problems with high-dimensional output vector as well as input vector. 
% There are two PLS algorithms. PLS-1 \citep{helland_structure_1988} for univariate responses ($q=1$) and PLS-2 for multivariate ($q>1$) described in \cite{breiman_predicting_1997}. 
Let $Y$ be $ n  \times q$ matrix of observed outputs and $ X $ be an $ n \times p$ input matrix. PLS finds a projection directions that maximize covariance between $X$ and $Y$, the resulting projections $U$ and $T$ for $Y$ and $X$, respectively are called  score matrices, and the projection matrices  $P$ and $Q$ are called loadings. The $X$-score matrix $T$, has $L$ columns, one for each ``feature" and $L$ is chosen via cross-validation. The key principle is that $T$ is a good predictor of $U$, the $Y$-scores. This relations are summarized by the equations below
\begin{align*}
	Y &  =  U Q+ E \\
	X & =  T P + F
\end{align*}
Here $Q$ and $P$ are orthogonal projeciton (loading) matrices and $T$ and $U$ are $n\times L$ are projections of $X$ and $Y$ respectively.

Originally PLS was developed to deal with the problem of collinearity in observed inputs. Although, principal component regression also addresses the problem of collinearity, it is often not clear which components to choose. The components that correspond to the larges singular values (explain the most variance in $X$) are not necessarily the best ones in the predictive settings. Also ridge regression addresses this problem and was criticized by  \citep{fearn_misuse_1983}. Further ridge regression does not naturally provide projected representations of inputs and outputs that would make it possible to combine it with other models as we propose in this paper. Thus, PLS seem to be the right method for high-dimensional problems when the goal is to model non-linear relations using another model. 

In the literature, there are two types of algorithms for finding the projections \citep{manne_analysis_1987}. The original one proposed by \cite{wold_collinearity_1984} which uses conjugate-gradient method \citep{golub2013matrix} to invert matrices. The first PLS projection $p$ and $q$  is found by maximizing the covariance between the $X$ and $Y$ scores \maxdisp{p,q}{\left(Xp\right)^T\left(Yq\right)}{||p||=||q||=1.} Then the corresponding scores are
\[
t = Xp,\text{   and   } u = Yq
\]
We can see from the definition that the directions (loadings) for $Y$ are the right singular vectors of $X^TY$ and loadings for $X$ are the left singular vectors. The next step is to perform regression of $T$ on $U$, namely $U = T\beta$.  The next column of the projection matrix $P$ is found by calculating the singular vectors of the residual matrices $(X - tp^T)^T(Y - T\beta q^T)$.

The final regression problem is solved $Y = UQ = T\beta Q = XP^T\beta Q $. Thus the PLS estimate is
\[
\beta_{\mathrm{PLS}} = P^T\beta Q.
\]

\cite{helland_partial_1990} showed that PLS estimator can be calculated as 
\[
	\beta_{\mathrm{PLS}} = R(R^TS_{xx}R)^{-1}R^TS_{xy}
\]
where $R = (S_{xy},S_{xx}S_{xy},\ldots,S_{xx}^{q-1}S_{xy})$,
\[
	S_{xx} = \dfrac{X^T(I-{\bf 11^T}/n)X}{n-1},
\]
\[
	S_{xy}  = ave ( y \bm x ). 
\]

Further, \cite{helland_structure_1988} proposed an alternative algorithm to calculate the parameters\\
For $K = 1 , \ldots, p$, set $ y_0 = y ,   \bm x_0 = \bm x $. Let  $ (y, \bm x )$ be centered and standardised. Given $ V = ave ( \bm x \bm x^T  ) $ and $ \bm s = ave ( y \bm x ) $,
for $K= 1 $ to $p$ do 
\begin{enumerate}
\item $  \bm s_k = V^{K-1} \bm s  $ 
\item $ \hat{y}_K = OLS ( y \; {\rm on} \; ( \bm s_k^T \bm x )_{k=1}^K ) $
\end{enumerate} 

\cite{stone_continuum_1990} discuss PLS and cross-validation.

\paragraph{Non-Linear Models and Ridge Functions}
% http://pages.cs.wisc.edu/~deboor/MAIA2013/slides/pinkus.pdf
% http://www2.math.technion.ac.il/~pinkus/papers/acta.pdf
The main building component of a deep learning model is a ridge function, which finds one dimensional-structure in the data \citep{pinkus2015ridge}. Deep learning is then a composition of several ridge functions to identify non-linear low-dimensional structure present in high-dimensional data.  

A ridge function is a multivariate function $f: \mathbb{R}^n \rightarrow \mathbb{R}$ of the form $f(x) = g(w^Tx) $ where $g$ is univariate function and $x,w \in \mathbb{R}^n$. The non-zero vector $w$ is called the direction. The name ridge comes from the fact that the function is constant along directions that are orthogonal to $w$, for a direction $\eta$ such that $w^T\eta = 0$, we have
\[
f(x+u) = g(w^T(x+u)) = g(w^Tx) = f(x)
\]
Arguably, this is the simplest multivariate function. \cite{klusowski2016risk} provided bounds for the bias when noisy data is modeled by smooth ridge functions. 

A generalization of a ridge function is a function of the form
\[
f(x) = g(Wx),
\]
where $W$ is a non-zero $d \times n$ matrix of weights and univariate function $g$ is applied element-wise. Deep learning is nothing but a composition of such generalized ridge functions.

Ridge functions are used in high-dimensional statistical analysis. For example, projection pursuit algorithm approximates the input-output relations using a linear combination of ridge functions \citep{friedman1981projection,huber1985projection,jones1992}
\[
\phi(x) = \sum_{i=1}^{r}g_i(w_i^Tx),
\]
where both the directions $w_i$ and functions $g_i$ are variables and $w_i^Tx$ are one-dimensional projections of the input vector. Projection pursuit regression is a dimensionality reduction technique. The vector $w_i^Tx$ is a projection of the input vector $x$ onto a one-dimensional space and $g_i(w_i^Tx)$ can be though as a feature calculated from data. Projection pursuit algorithm is an iterative procedure that finds a predictive function iteratively by starting with $\phi_0 = 0$ and proceeds
\[
w_{i+1} \in \arg\min_w \lVert \phi_{i+1}  - y\rVert_p;~~\phi_{i+1} = \phi_i + g_{i+1}(w^Tx).
%  = f_n + \sigma_{n+1}(w^{(n+1)T}x), \lVert f_{n-1}  - f\rVert_p \mbox{ is minimum}
\]

\textit{Nonlinear PLS} was first considered by \cite{kramer1991nonlinear} and \cite{malthouse1997nonlinear} who proposed projecting (non-orthogonally) the input variables into an $L$-dimensional surface  which is parametrized by a neural network. Those projections, which are equivalent of scores vectors in PLS are then used for the prediction of the output. \cite{rosipal2001kernel} develops a kernel PLS algorithm which is a non-linear regression in high dimensions. First, they find a transformation of the input vector $x \rightarrow \Phi(X)$ and then use PLS to map $\Phi(X)$ to $y$. The features $\Phi(x)$ are calculated as using a kernel function that depends on both $x$ and $y$. It can be shown that kernel regression is similar to a radial basis function network \cite{haykin_neural_1998}.

\paragraph{Deep learning} DL overcomes many classic drawbacks by \textit{jointly}
estimating non-linear function that does dimensionality reduction and prediction using information on $Y$ and $X$ as well as their relationships, and by using $L>2$ layers.

A deep learning is simply a composition of ridge functions. The prediction rule is embedded into a parameterised deep learner, a composite of univariate semi-affine functions, denoted by
$G_W$ where $ W = [ w^{(1)}, \ldots , w^{(L)} ] $ represents the weights of each layer of the network.  A deep learner takes the form of a composition of link functions
$$
G_W = G_1 \circ \dots \circ G_L  \; {\rm where} \; G_L = g_L ( w_L x + b_L) 
$$
where $ g_L $ is a univariate link function.

If we choose to use nonlinear layers, we can view a deep learning routine as a hierarchical nonlinear factor model or, more specifically, as a generalised linear model (GLM) with recursively defined nonlinear link functions,  see \cite{polson_deep_2017} and \cite{tran2020bayesian} for further discussion.

\cite{diaconis1984nonlinear} use  nonlinear functions of linear combinations. The hidden factors $Z_i= W_{i1}X_1 + \ldots + W_{ip}X_p$ represents a data reduction of the output matrix $X$. The model selection problem is to choose how many hidden units. We will show that PLS scree-plot can be used as a guide to find $L$. Predictive cross-validation can also be sued.

We now turn to estimation of single index models with non-Gaussian regressors. 

\subsection{Projection to Single Index Model}

\cite{brillinger_generalized_2012} considers the single-index model with non-Gaussian regressors where  $(Y,X) $ are stochastic with conditional distribution
$$
Y \mid X \sim  N(g ( \alpha + \beta X ),~\sigma^2).
$$
Here $\beta  X  $ is the single features found by data reduction from high dimensional $ X$. 
Let $ \hat{\beta}_{OLS} $ denote the least squares estimator which solved $X^TY = X^TX\beta$.  By Stein's lemma, 
$$
\cov(Y,X)  = \beta \cov ( g ( \alpha + \beta X ) , \alpha + \beta X ) \var(X) / \var( \alpha + \beta X ) 
$$
Then $ \hat{\beta} $ is consistent as
$$
\hat{\beta}_{OLS} = \cov( Y,X) / \var(X) \rightarrow k \beta \;  \; {\rm where} \; \; k = \cov ( g ( \alpha + \beta X ) , \alpha + \beta X ) \var(X) / \var( \alpha + \beta X ) 
$$
Hence,  $ \hat{\beta}_{OLS} $ estimator is proportional to $ \beta $. We can also non-parametrically estimate $g(u)$ by plotting 
$ ( \hat{\beta} x_j , y_j  ) , ~j = 1 , \ldots n $ and smoothing $y_j$ values with $  \hat{\beta} x_j $ near $ u $. 

Hence, when the $X$s are Gaussians  and independent of the error, we have the relationship $ \cov(Y,X) = k \beta \var(X) $. This relationship follows from the weaker assumption
$$
E( Y|X ) = g ( \alpha + \beta X ) 
$$
Hence, this approach can be applied to binary classification with $ Prob ( Y=1 | X ) = g( \alpha + \beta X ) $ and other models such as survival models.

\cite{Naik2000} apply this to partial least squares. \cite{hall1993} show that when dimension of $X$ is high then most nonlinear regressions are still nearly linear. 

The following theorem extends Brillinger's result to deep learning with ReLU activation.
\begin{theorem}
Given a deep learning model of the form
$$
E( Y|X ) = G_1 \circ \dots \circ G_{L}(X)
$$
where $G_1(Z) = B_1Z$ and $G_l(Z) = \max(B_lZ, 0)$ for $ l=2,3,..., L$. $B_L\in\mathbb{R}^{D\times p}$ and $B_l \in \mathbb{R}^{D\times D}$ for $l=1,2,..., L-1$. $X\in\mathbb{R}^p$ and $Y\in\mathbb{R}^D$.

Then we can consistently estimate $ ( B_1 , \ldots , B_L ) $ recursively using OLS.
\end{theorem} 
\begin{proof}
Define $Z_L := X, \; Z_0 :=Y$ and $Z_{l-1} := G_l(Z_{l})$ for $l=1,2,3,..., L$. First,  we write the deep learning model as
$$
E( Y|X ) =  G^{(L)} ( B_{L}  Z_L  ) \; \; {\rm where} \; \;G^{(L)} = G_1 \circ \dots \circ G_{L-1} \circ \max(\cdot, 0).
$$
Running OLS of $ Y $ on $X$, allows using Brillinger's result to consistently estimate $ B_L $. That is, we get $\hat B_L \approx k_L B_L$ where $k_L$ is a diagonal matrix. Then construct
$$
\hat Z_{L-1} = \max(\hat B_L X, 0) \approx \max(k_L B_L X, 0) = k_L Z_{L-1}
$$
and run OLS of $Y$ on $\hat Z_{L-1}$. Note that now the corresponding model becomes
$$
E(Y|X) = G^{(L-1)}(B_{L-1}Z_{L-1}) = G^{(L-1)}\left(k_{L}^{-1}B_{L-1}\cdot k_LZ_{L-1}\right)\; \; {\rm where} \; \;G^{(L-1)} = G_1 \circ \dots  \circ G_{L-2} \circ \max(\cdot, 0).
$$ 
According to Brillinger, the resulting OLS estimate $\hat B_{L-1}$ satisfies $\hat B_{L-1} \approx k_{L-1}k_{L}^{-1} B_{L-1}$. And correspondingly, 
$$\hat Z_{L-2} = \max(\hat B_{L-1}\hat Z_{L-1}, 0) \approx \max(k_{L-1} B_{L-1} Z_{L-1}, 0) = k_{L-1} Z_{L-2}.$$
Therefore, by induction,  we can estimate $ ( B_1 , \ldots , B_L ) $ recursively solely using OLS.
\end{proof}
\cite{Erdogdu2016} proposes an efficient algorithm for calculating the proportionality constant $k$.
The proportionality was first observed by Fisher for logistic regression with Gaussian design.
\cite{antoniadis2004bayesian} provide techniques for Bayesian estimation of the single-index models using B-spline approximate $g$ functions. 

\cite{xia2009adaptive} consider the generalized single index model
\[
Y = g(B X) + \epsilon
\]
where $B\in \mathbb{R}^{D\times p}$ is orthogonal matrix.  $B X$ is known as the effective dimension reduction (EDR)

Stein's lemma can also  be extended to scale mixtures of Gaussians (\cite{gron2012optimal}) and a similar analysis holds.

\section{Bayesian Shrinkage via Partial Least Squares}

\subsection{Ridge, PCR and PLS} 

Linear prediction rules without shrinkage, 
$ \hat{Y} = X\beta  $ where $ \beta = ( X^T X )^{-1} X^T Y  $, have the caveat of high variance due to the typical ill-conditioned nature of the 
$ ( X^T X ) $-design matrix.   Thus shrinkage methods are been developed. 

RR, PCR and PLS are all operationally the same.  They bias coefficient vector away from directions in which the predictors 
have low sampling variability--- or equivalently, sway from the "least important" principal components of $X$ 
\begin{enumerate}
	\item  With high collinearity in which variance dominates the bias, the performance of Ridge, PCR and PLS tend to be quite similar. The three methods shrink heavily along those directions of small predictor variability. They are also invariant under rotation, but not scaling.
	\item  Unlike Ridge and PCR, PLS biases the coefficient vector towards directions in the predictor space that preferentially predict the high spread directions in the response space. It not only shrinks the OLS in some eigen-directions, but expands it in others (related to $K$-th eigenvalue).
	\item For PLS, the scale factors along each of the eigen-directions are not linear in the response values (smooth but not monotonic). They also depend on the OLS solution (not the length, but just the relative values). 
	\item PLS typically uses fewer components to achieve the same overall shrinkage as PCR, generally about half as many components. If the number of component being equal, PLS has less bias and greater variance, compared with PCR.
	\item For multi-response problems, two-block PLS or the multi-response ridge doesn't dramatically do better than the corresponding uni-response procedures applied separately to the individual responses.
\end{enumerate}
As with any automatic procedure, there are caveats where certain empirical datasets can load on the eigen-rvectors with the smallest eigen-values,
see, for example, \cite{fearn_misuse_1983}, \cite{polson2012}. PLS can be viewed as an  data pre-processing step before implementing deep learning. 

\subsection{Bayesian Shrinkage Interpretation}

We will rely heavily on the Bayesian shrinkage interpretation of PCR and PLS due to Frank and Friedman (1993). \cite{polson2010,polson2012} provide a general theory of global-local shrinkage and, in particular,  
analyze $g$-prior and horseshoe shrinkage. PLS behaves differently from standard shrinkage rules as it can shrink away from the origin for certain eigen-directions. 

First, start with the SVD decomposition of the input matrix: $ X = U D W^T $. If $ n > p $ then of full rank.
Here $ D = diag ( d_1 , \ldots , d_p ) $ nonzero ordered $ e_1 > \ldots > e_p $ singular values.
Then $W = ( w_1 , \ldots , w_p ) $ is the matrix of eigenvectors for $ S = X^T X $. 
We can then transform the original first layer to an 
orthogonal regression, namely $ y = ( UD ) W^T \beta $ with corresponding  OLS estimator $ \hat{\alpha} = ( Z^T Z )^{-1} Z^T y = D^{-1} U^T y $.

\paragraph{Ridge, PCR and PLS Algorithm}
PCR has a long history in statistics. This is 
an unsupervised approach to dimension reduction (no $y$'s). 
Specifically, we first  center and standardise $ (y, \bm x )$.

Then, we provide an SVD  decomposition of 
$$ V = ave ( \bm x \bm x^T ) $$  
This find the eigen-values $ e_j^2 $ and eigenvectors arranged in descending order, 
so we can write 
$$ V = \sum_{j=1}^p  e_j^2 {\bm v}_j {\bm v}_k^T .
$$ 
This leads to a sequence of regression models $(\hat Y_0, ..., \hat Y_L)$ with $ \hat Y_0 $ being the overall mean and 
$$
\hat{Y}_L = \sum_{l=0}^L (ave ( w_l^T \bm x ) / e_l^2 ) \bm v_l^T \bm x 
$$
Therefore, PLS finds  "features"  $Z_K = \{\bm v_k^T \bm x\}_{k=0}^K = \{\bm f_k \}_{k=0}^K $.  

The key insight is that all of the estimators are of the from 
$$ 
 \hat Y^M = \sum_{j=1}^L  f_j^M  \hat{\alpha}_j \bm v_j^T \bm x 
$$
where $f_j$ are scale factors. $M$ denotes method (e.g. RR, PCR, PLS). $L$ is the rank of $\bm V$ (number of nonzero $e_k^2$). 
For PCR, the scale factors are  $f_j =1 $  for top $L$ eigenvectors  \cite{frank_statistical_1993} 
\begin{align*}
	f_j^{RR} &= e_j^2/(e_j^2 + \lambda ), \text{where $\lambda$ is a fixed regularization parameter}\\
	f_j^{PCR} &= \begin{cases}
	1, & e_j^2 \geq e_L^2\\
	0, & \text{otherwise}
	\end{cases}\\
f_j^{PLS} &= \sum_{k=1}^{K} \theta_k e_j^{2k} \; ,  \; {\rm where} \;  \theta = w^{-1}\eta ,  \;  \eta_k = \sum_{j=1}^p \hat{\alpha}_j^2 e_j^{2(k+1)} . 
\end{align*}
The corresponding shrinkage factors for RR and PCR are typically normalized so that they give the same overall shrinkage so that the length of solution vector are the same ($|\hat\beta_{RR}| = |\hat\beta_{PLS}|$).

This scale factors provide a diagnostic plot: $f_j $.  If any $ f_j > 1 $ then one can expect 
supervised learning (a.k.a. PLS with $Y$'s influence the scaling factors) will lead to different  predictions than   unsupervised learning (a.k.a. PCR with
solely dependent  on $X$).  In this  sense, PLS is an optimistic procedure in that the goal is  
to maximise the explained variability of the output in sample with the  hope of generalising well out-of-sample.

\paragraph{Model selection (a.k.a. dimension reduction)}
The goal of PCR is to minimize predictive MSE
$$
\hat{L} = \arg \min_L ave( y - \hat{y}_L )^2 
$$
The choice of $L$ is determined via predictive 
 cross-validation. The $l$th model is simple regression of $ y $ on $ f_L = \bm v_l \bm x , l = 1 , \ldots L $. 
\cite{mallows_comments_1973} $ C_p $ and $ C_L $ provide the relationship between shrinkage and model selection.

\paragraph{Dropout} This is a model selection technique designed 
to avoid over-fitting in deep learning.  This is done by  removing input dimensions in $X$ randomly with a given probability $p$. 
For example, suppose that we wish to minimise MSE,  $ \|Y-\hat{Y}\|^2_2$, then, when marginalizing over the randomness, we have a new objective
$$
{\rm arg \; min}_W \; \mathbb{E}_{ D \sim {\rm Ber} (p) } \Vert Y - W ( D \star X ) \Vert^2_2\,,
$$
This is equivalent to, with $ \Gamma = ( {\rm diag} ( X^\top X) )^{\frac{1}{2}} $, 
$$
{\rm arg \; min}_W \;  \Vert Y - p W X \Vert^2_2 + p(1-p) \Vert \Gamma W \Vert^2_2\,,
$$
Hence, this is equivalent to a Bayesian ridge regression with a $g$-prior as an objective function and reduces the likelihood of over-reliance on small sets of input data in training.

\section{Deep Learning via Partial Least Squares}
Partial least squares algorithm finds projections of the input and output vectors $X = TP+F$ and $Y = UQ+E$ in such a way that correlation between the projected input and output is maximized.  Our DL-PLS model will introduce nonlinearity $U  = G(T)$ by assuming that $U$ is a deep learner of $T$. From Brillinger's result and Theorem 1 we see that linear PLS calculates $T$ and $P$ for arbitrary $G_L$. 
\begin{align*}
 Y &  =  U Q+ E \\
 U & = G(T) \\
T & =  X P^T 
 \end{align*}
where $G$ is a deep learner. 
Here $ X = T P + F $ is inverted to $ T = X P^T $ as $ P ^T P = I $.

\begin{figure}[H]
	\centering
	\includegraphics[width=\linewidth]{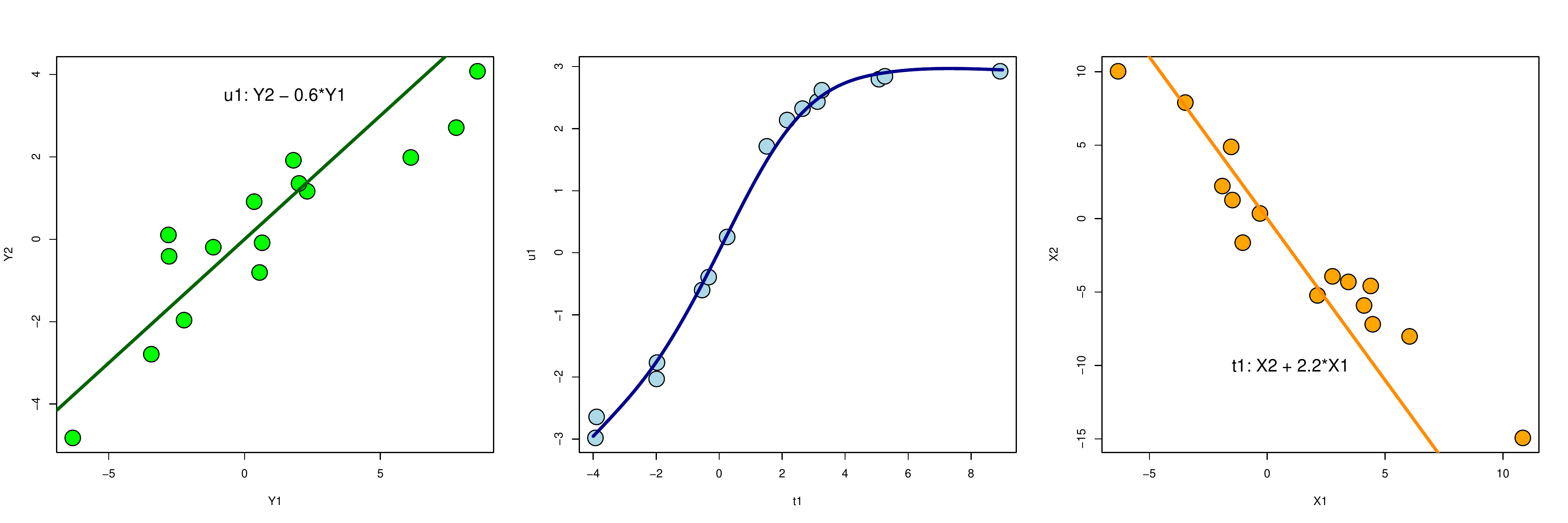}
	\caption{Relationship between output-input $Y$ and $X$\\
		$Y = U Q + E, \; \; U = G ( T) , T = X P^T $}\label{xyut}
\end{figure}

Figure \ref{xyut} shows the relationship between $X$ and $Y$, via the link between score vectors. $t_1$ and $u_1$ are linear combinations of $X$ and $Y$ respectively (left and right subplots). The relationship between $t_1$ and $u_1$ is nonlinear (middle subplot).

Although we use a composite model DL-PLS, our estimation procedure is two-step. We first estimate the score matrices and then estimate parameters of the deep learning function $G$. This two step process is motivated by the \cite{brillinger_generalized_2012}. The results of Brillinger guarantee that  matrices $P, Q$ are invariant (up to proportionality) to nonlinearity, even when the true relationship between  $Y$-scores and $ X $-scores is nonlinear.

There are many choices for the network architecture of the deep learner $G$. 
\paragraph{PLS-ReLU}
For example, $G$ is a simple feed-forward ReLU neural network, for which sparse Bayesian priors are useful to improve the generalization (\cite{polson2018posterior}). $U$ and $T$ are $n\times L$ matrices, 
\begin{align*}
	T &= Z_0\\
	Z_1 &= \max(Z_0W_1+b_1,0)\\
	U &= Z_1W_2+b_2
\end{align*}
The weights $W_1$ and $W_2$ are to be learned.
 
\paragraph{PLS-Autoencoder}
As the input $T$ and output $U$ have the same dimensions, one can consider using an autoencoder network. We build on the architecture of \cite{zhang2019improved}. 
\begin{align*}
	T &= Z_0\\
	Z_1 &= g_1(Z_0W_1+b_1)\\
	U &= g_2(Z_1W_2+b_2)
\end{align*}
where $Z_1$ is a $n\times l_1$ matrix which acts as a lower dimensional intermediate hidden layer and $l_1 \ll l_0 = L$.
\paragraph{PLS-Trees}
One approach is Eriksson, Trygg and Wold. The first score vector, $t_1 = Xw$ where $w= X^Ty/\Vert X^Ty\Vert^2$. After sorting all observations from maximum to minimum score value along $t_1$, they divides $X$ and $Y$ in two parts to minimize 
$$
(1-b)\cdot \left(a\cdot\frac{Var(y_1)+Var(y_2)}{Var(y)} + (1-a)\cdot \frac{Var(t_1)+Var(t_2)}{Var(t)}\right) + b\cdot\frac{(N_1-N_2)^2}{(N_1+N_2)^2}
$$
where $Var$'s are appropriate variances and $N$'s are sample sizes. The last term $\frac{(N_1-N_2)^2}{(N_1+N_2)^2}$ acts as a regularization penalty. The hyperparameters $a$ and $b$ are chosen by cross-validation.

\paragraph{PLS-GP}
Given a new predictor matrix $X_{*}$ of size $N_*\times p$, the same projection $P$ produces the corresponding $T_{*} = [t_{1,*}, ..., t_{L, *}]$. We can use Gaussian process regression to predict $U_{*} = [u_{1,*}, ..., u_{L, *}]$ from $T_{*}$ as follows
\begin{align*}
	T_{*} &= X_{*} P^T\\
	\hat u_{k, *} &= g_{k,*}(t_{k, *}), k=1,2,...,L.
\end{align*}
where $g_{k,*}$'s are the Gaussian process regression predictors. 

\begin{equation*}
\begin{bmatrix}
u\\
u_*
\end{bmatrix}\sim N\left(
\begin{bmatrix}
g\\
g_*
\end{bmatrix},
\begin{bmatrix}
K & K_*\\
K_*^T & K_{**}
\end{bmatrix}\right)
\end{equation*}
where $K = K(t, t)$ is $N\times N$, $K_* = K(t, t_*)$ is $N\times N_*$, and $K_{**} = K(t_*,t_*)$ is $N_*\times N_*$. $K(\cdot,\cdot)$ is a kernel function. The conditional mean, $g_*$, is given by
\begin{align*}
%p(u_* \mid t_*, t, u) &= N\left(u_* \mid g_*, \Sigma_*\right) \\
g_*(t_*) &= g(t_*) + K_*^TK^{-1}(u - g(t)). 
%\Sigma_* &= K_{**} - K_*^TK^{-1}K_*
\end{align*}
Then prediction of $Y$ is
\begin{equation*}
	\hat Y_* = \hat U_{*} Q
\end{equation*} 
\cite{Fadikar_2018} use PCA to reduce the 57-dimensional output vector together with Gaussian process regression. Alternatively, one could perform partial least squares to reduce the dimensionality to provide a better calibration of the simulation-based model.
% Nonlinearity via $ g $ is a deep learner. Multi-layer version for $g$. 

To summarize,
\vspace{-0.15in}
\paragraph{Algorithm} 
Our overall procedure includes 5 steps
\begin{enumerate}
	\item Given input $X$, expand by including modeler-defined features to get an augmented set of inputs $\phi(X)$, see \cite{Hoadley2000}.
	\item Use SVD to find eigenvalues-eigenvectors of $V = ave(\phi(X)\phi(X)^T)$
	\item Use PLS to estimate weight matrices $Q,P$ and $Y$-scores $U$. Select dimensionality $L$ via cross-validation and scree-plot. Let $T$ denote $X$-scores
	\item Model $U = G(T)$ where $G$ is a deep learner. Use diagnostics plot to interpret the model.
	\item Run MCMC on model $Y = UQ + \epsilon$ where $U = G(T)$ and $T = XP^T$ to obtain predictive distribution $p(Y \mid X)$
\end{enumerate}

\section{Applications}
We use two simulated and two real-world data sets to study the performance of PLS-DL in estimating the input-output relations. In our first two simulated examples we can see that estimated $\beta_{\mathrm{PLS}}$ provides a mechanism to select the number of latent dimensions and to plot the input-output relations in lower dimensional space.

\subsection{Simulated Data: ReLU and tanh}
To illustrate Brillinger's method for identifying the activation function, $g$, we simulate data from two neural networks with ReLU and tanh activation functions. 
Asymptotically $ \hat{\beta}_{OLS} $ estimates $ \beta $ up to proportionality and estimates of  $g(u)$ can be obtained by plotting 
$ ( \hat{\beta} x_j , y_j  ) , j = 1 , \ldots n $ and smoothing $y_j$ values with $  \hat{\beta} x_j $ near $ u $. 

We set $n=50$ and $\beta \in \mathbb{R}^{20}$. Figure \ref{brillinger beta} plots the OLS $ \hat{\beta} $'s versus the true $ \beta$'s. We see clearly that we can estimate the $\beta$'s up to proportionality.
Figure \ref{brillinger_activation} plots  $ ( \hat{\beta} x_j , y_j  ) , j = 1 , \ldots n $  and show how we can recover ReLU and tanh fucntions albeit with noise. The scaling factor between OLS $\hat\beta$'s and the true $\beta$'s is estimated together with the activation function. That is, if $\hat\beta_{OLS} \approx k \beta$, then $\hat g(u) \approx g(\frac{1}{k}\cdot u)$.

\begin{figure}[H]
	\centering
	\includegraphics[width=0.75\linewidth]{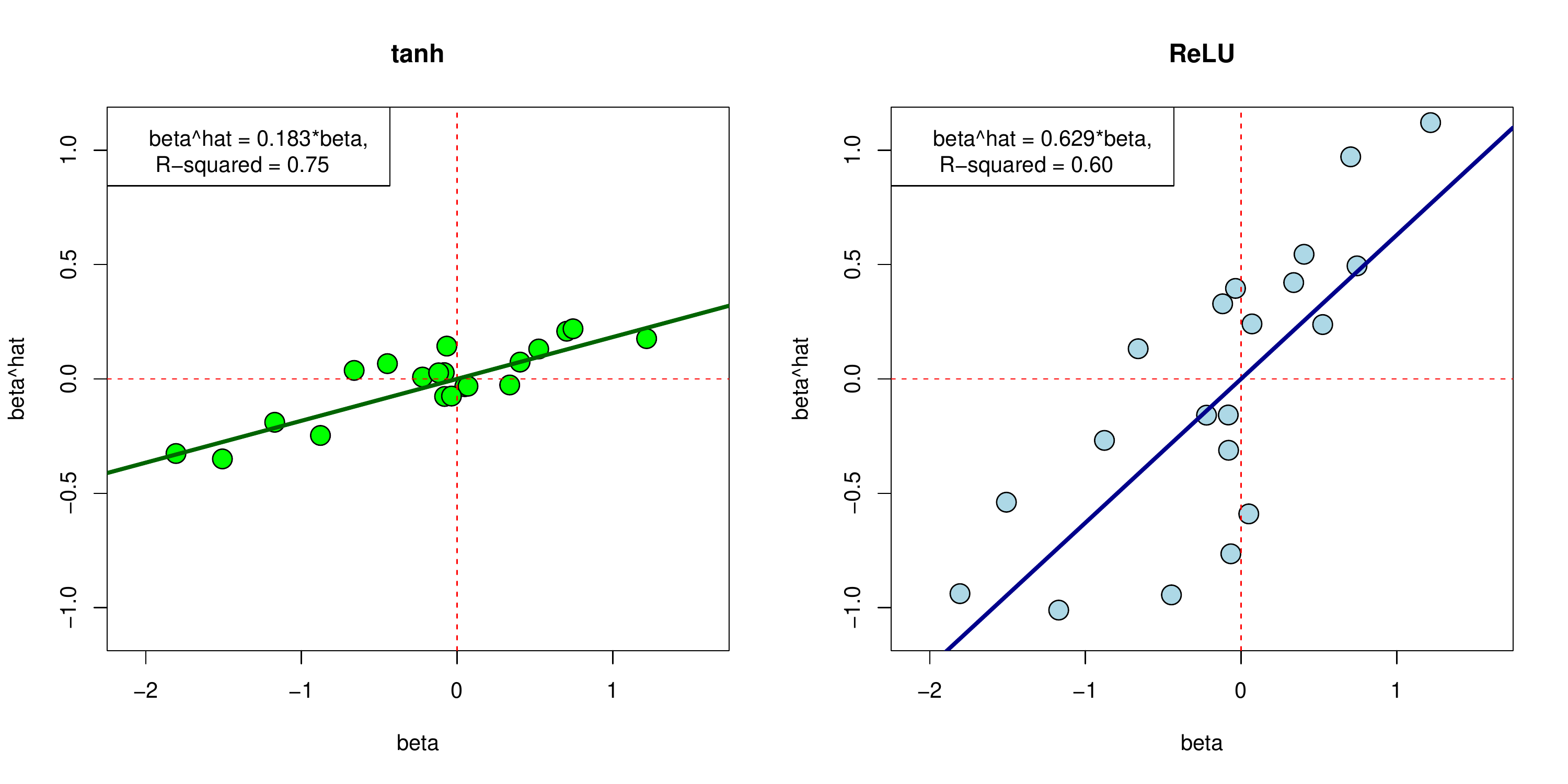}
	\caption{OLS $\hat \beta$ vs. true $\beta$: the regression model $\hat\beta \sim \beta$ is shown in the top-left corner.}\label{brillinger beta}
\end{figure}

\begin{figure}[H]
	\centering
	\includegraphics[width=0.75\linewidth]{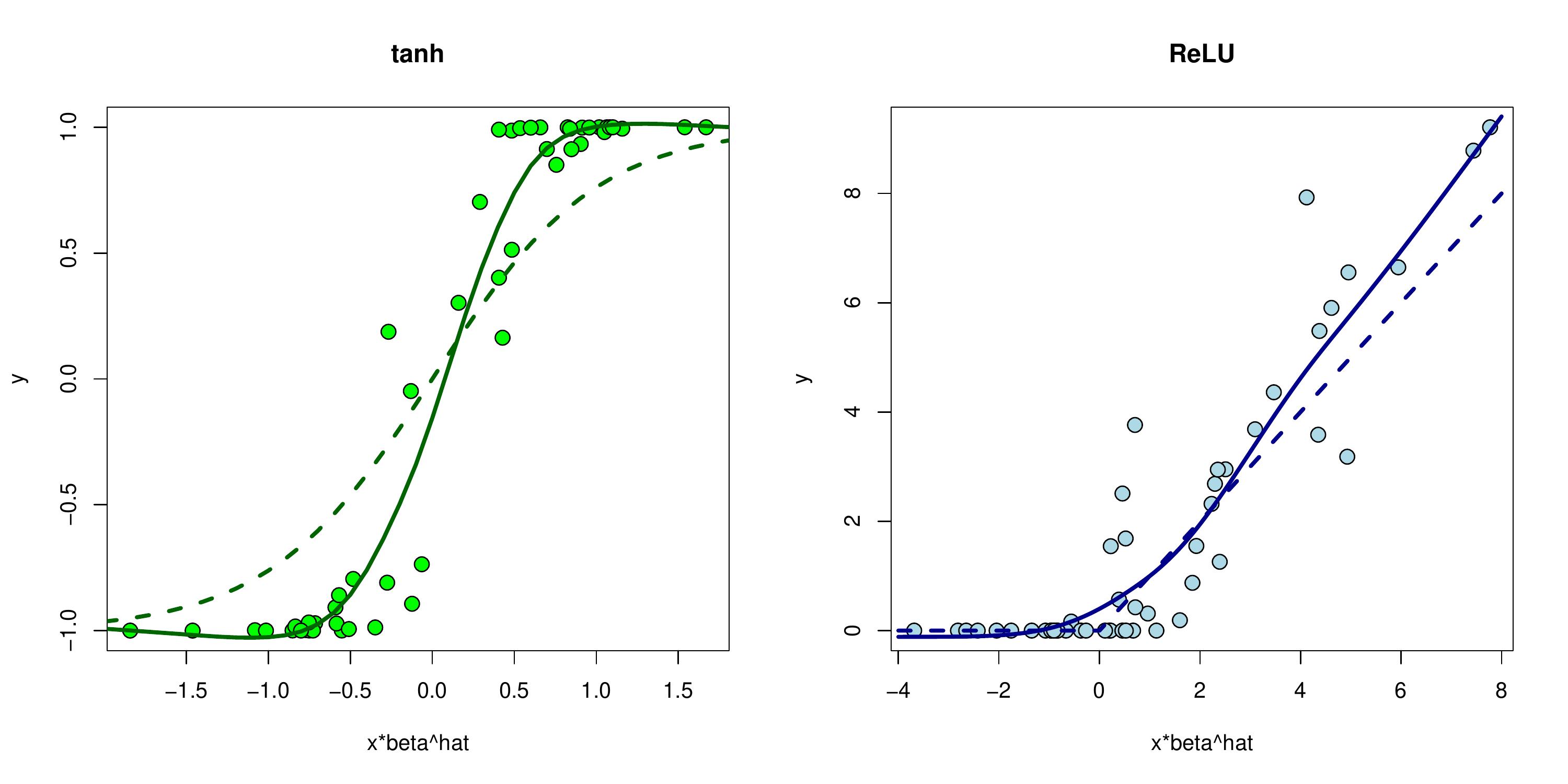}
	\caption{$y$ vs. $\hat\beta x$: the solid lines are estimated activation $\hat{g}$'s and the dashed lines are true $g$'s.}\label{brillinger_activation}
\end{figure}

\subsection{Simulated Data with Independent Predictors}
We start with the case when predictors are independent and PLS estimates are the same as OLS. We start with a non-linear function 
\[
y = \log |1+ Bx| + \epsilon
\]
where $B = (1,2,0,0)$ and $x\sim N(0,I_{4\times 4})$ and $Var(\epsilon) = 0.005$.
\begin{figure}[H]\label{fig:naik}
	\centering
	\includegraphics[width=\linewidth]{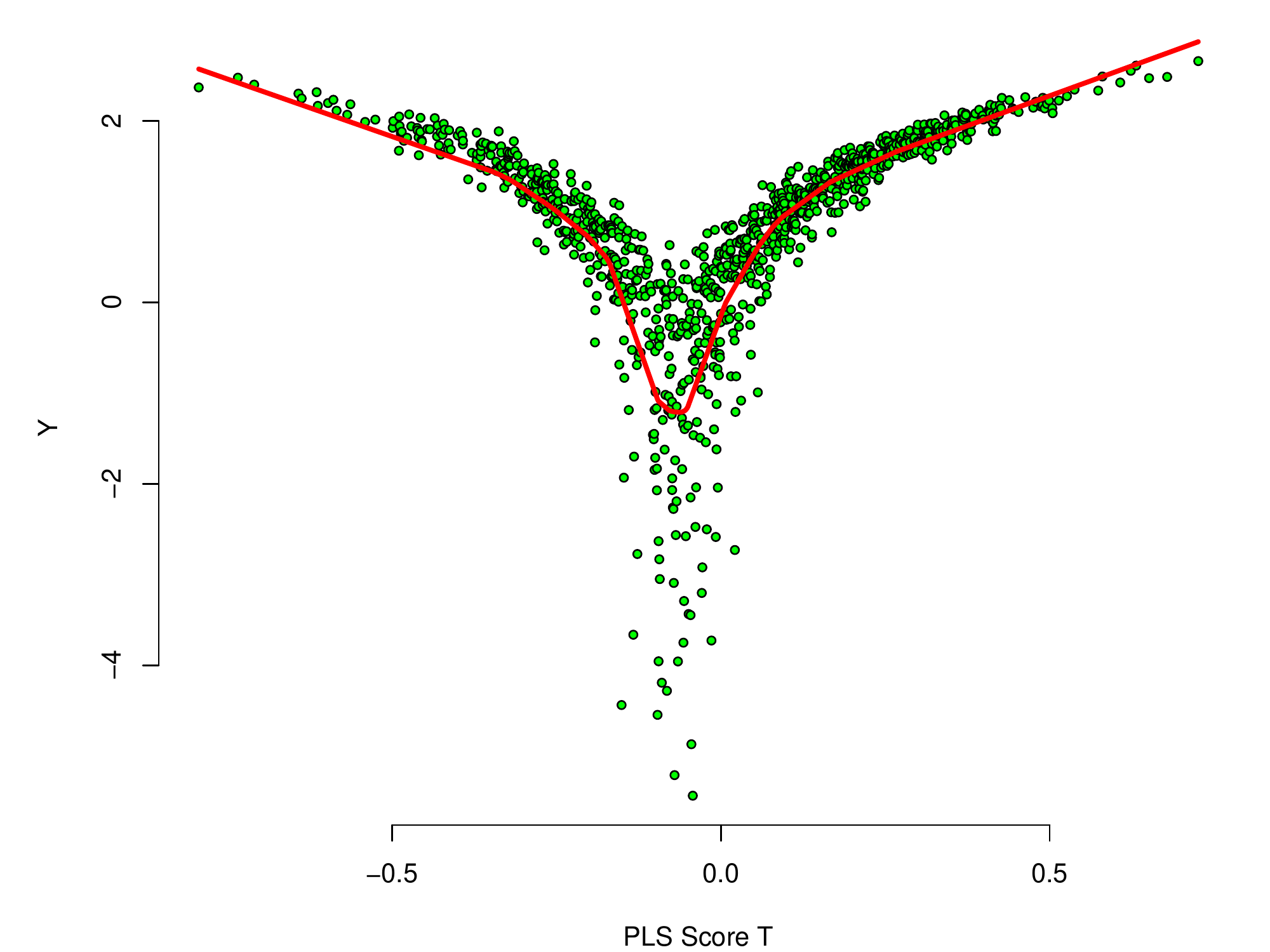}
	\caption{The training data set generated using the function $y = \log |1+ x^TB| + \epsilon$  (green) and curve fitted between first score and $y$ using neural network (red) }
\end{figure}

We further extend this example by adding another output 
\begin{align*}
	y_1 = &  B_1x + \epsilon_2	\\
	y_2 = & \log |1+ B_2x| + \epsilon_2
\end{align*}
where $B_1 = (0,0,2,2)$, $B_2 = (1,2,0,0)$ , and $Var(\epsilon) = \left(\begin{array}{cc}0.001 & 0.005 \\ 0.005 & 0.001
\end{array}\right)$

The estimated coefficients by the PLS algorithms are shown in Table \ref{tab:pls} and we can see that PLS estimated correctly $B_1$ since there is a linear relation between $B_1x$ and $y$ and the estimates of $B_2$ are proportional to the true values. 
\begin{table}[H] \centering 
\caption{Coefficients estimated by the PLS algorithm} 
\label{tab:pls} 
\begin{tabular}{@{\extracolsep{5pt}} cc} 
\toprule
$\hat \beta_1$ & $\hat \beta_2$\\ \hline
$0.003$ & $0.138$ \\ 
$0.001$ & $0.216$ \\ 
$2.002$ & $0.034$ \\ 
$2.004$ & $0.003$ \\  \bottomrule
\end{tabular} 
\end{table}

Then we fitted a neural network that takes two inputs, which are first two scores calculated by PLS and uses original $y$ as an output. The neural network correctly learns the identity relation between the first score and $y_1$ and also correctly learns the non-linear part of the relation as shown in Figure \ref{fig:naik2d}.

\begin{figure}[H]
\centering
\begin{tabular}{cc}
\includegraphics[width=0.45\linewidth]{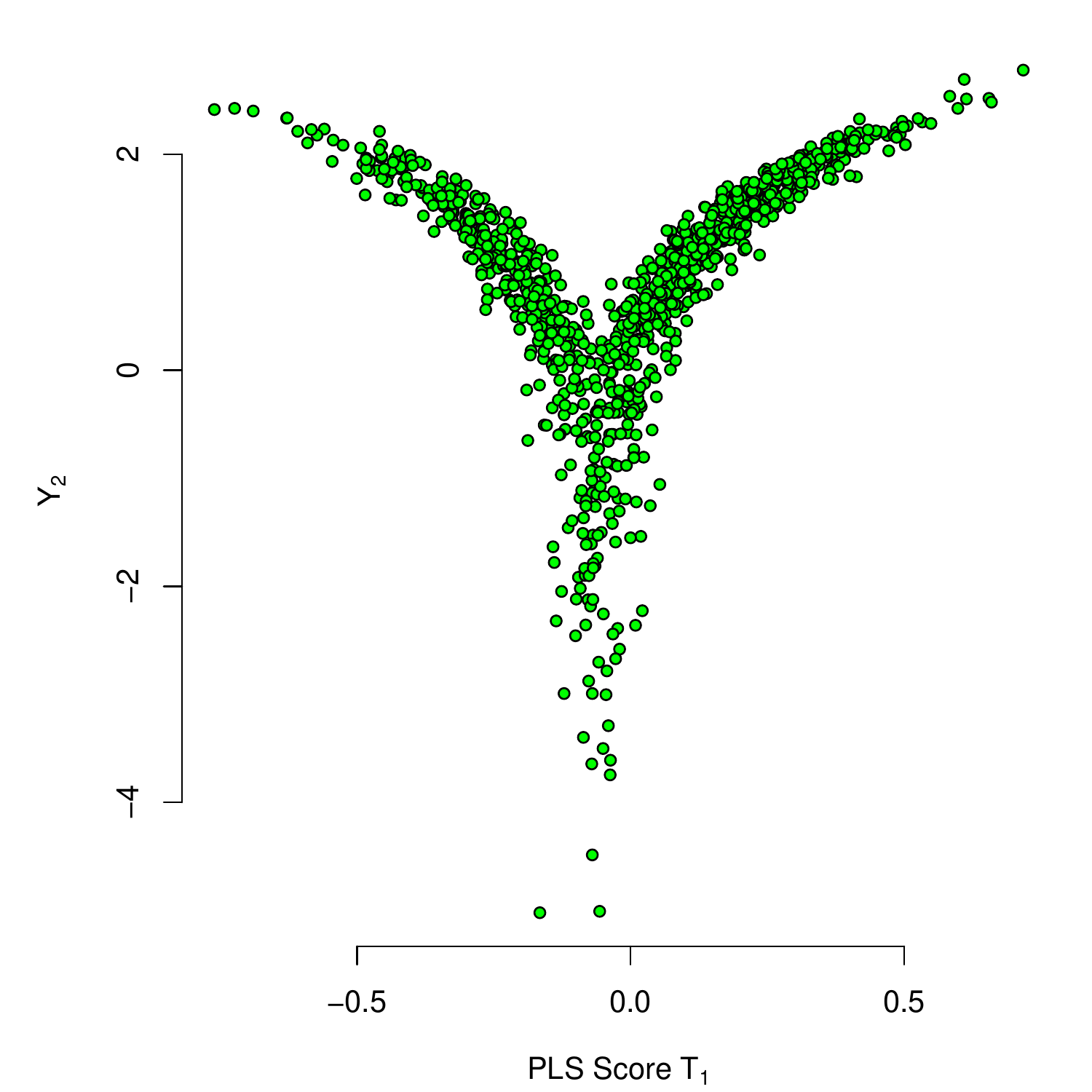} & \includegraphics[width=0.45\linewidth]{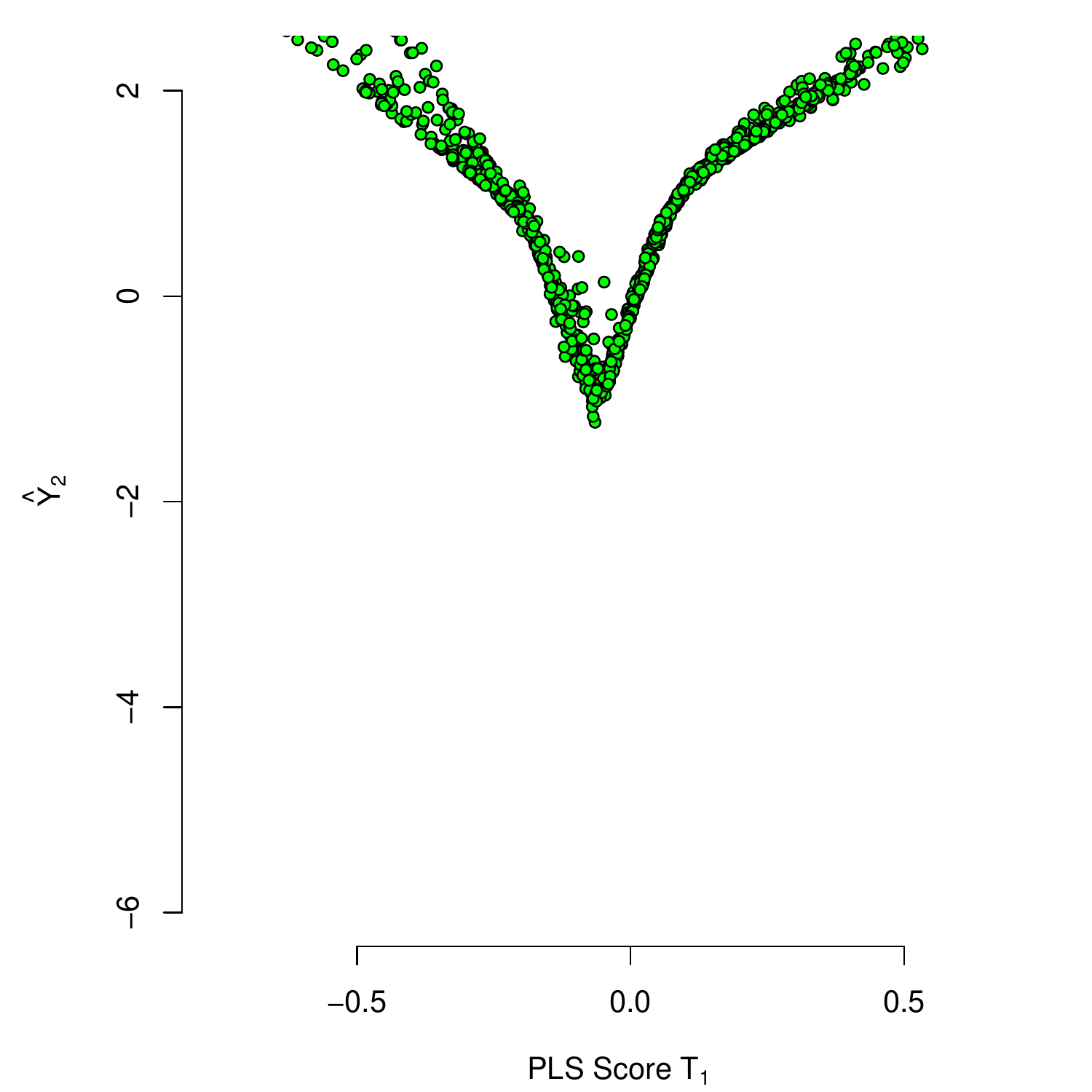}\\
(a) Training Data & (b) Fitted Data
\end{tabular}
\caption{Training versus PLS-DL fitted data}\label{fig:naik2d}
\end{figure}

\subsection{Orange Juice Data}

\cite{tenenhaus_pls_2005a} applies PLS to the orange juice data. There are $q$ hedonic judgements and $p$ product characteristics. In the orange juice dataset, we have 
$q=96$ judges and $ p=16$ predictors. In the so-called external analysis, a principal component analysis is carried out on $X$.  Then
the $Y$ are regressed on the selected principal components after predictive cross-validation has been used to select the number of 
principal components.  For example, with $L=2$. The predictive rule underlying PLS is achieved as a superposition (as in deep learning) of the 
PLS components $(t_1 , t_2 ) $  of the PLS regression of $Y$ on $X$ and the coefficient vector $ (p_1 , p_2 ) $ of 
the regression of $X$ on $(p_1 , p_2 ) $.  Characteristics are summarised by $ (t_1 , t_2 ) $  and $ \hat{Y}_k = c_{1k} t_1 + c_{2k} t_2 $.

Figure \ref{diagnosis} shows some diagnostic plots for our DL-PLS model. The left panel represents the reduction in unexplained variance by increasing the number of components. This screeplot suggests that first three components can explain almost all variation. The middle panel is a biplot for $X$ with respect to the first two components, using SVD of $X$. Samples and variables are displayed together, where the distances measure the similarity between samples/variables. For example, one observes that glucose and fructose are highly correlated. And another group of variables, including sucrose, ph1, ph2, citric, smell.int, taste.int, odor.typi, acidity, bitter and sweet, are also close to each other. The right panel of Figure \ref{diagnosis} is the result of applying Brillinger on $T$ and $U$. The nonlinear function $g$ is fitted with a smoothing spline. 

\begin{figure}[H]
	\centering
	\includegraphics[width=0.4\linewidth]{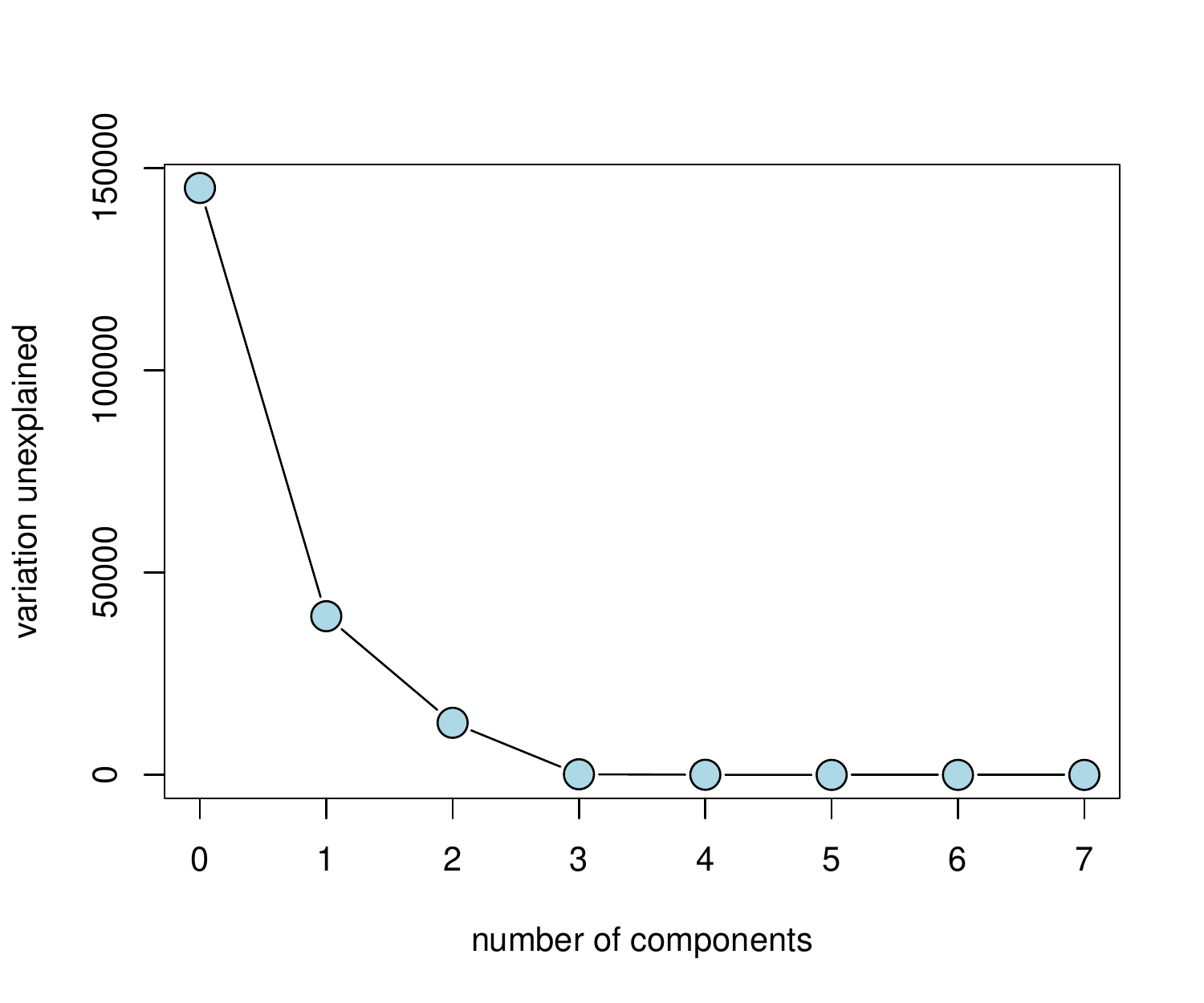}
	\includegraphics[width=0.4\linewidth]{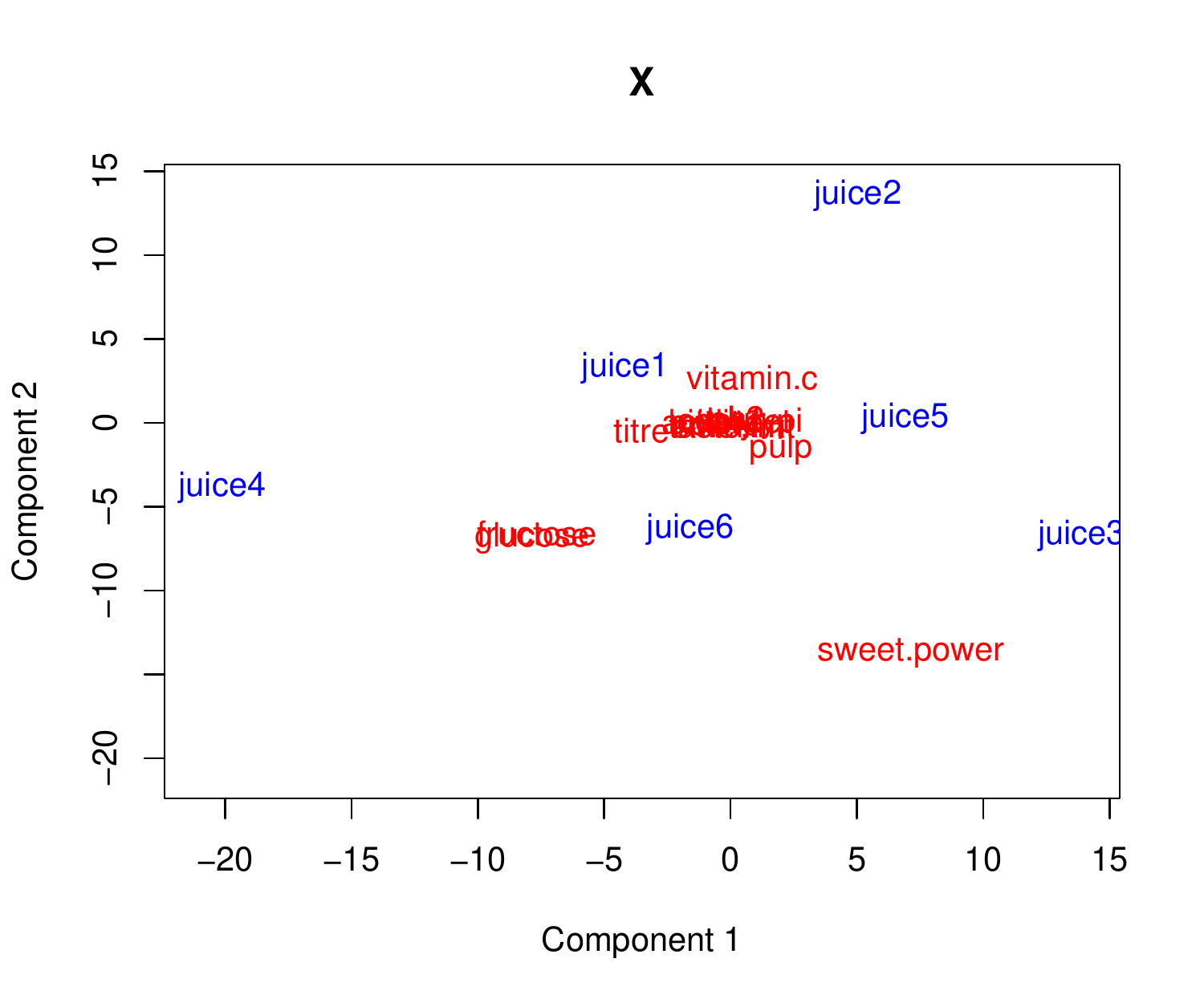}
	\includegraphics[width=0.4\linewidth]{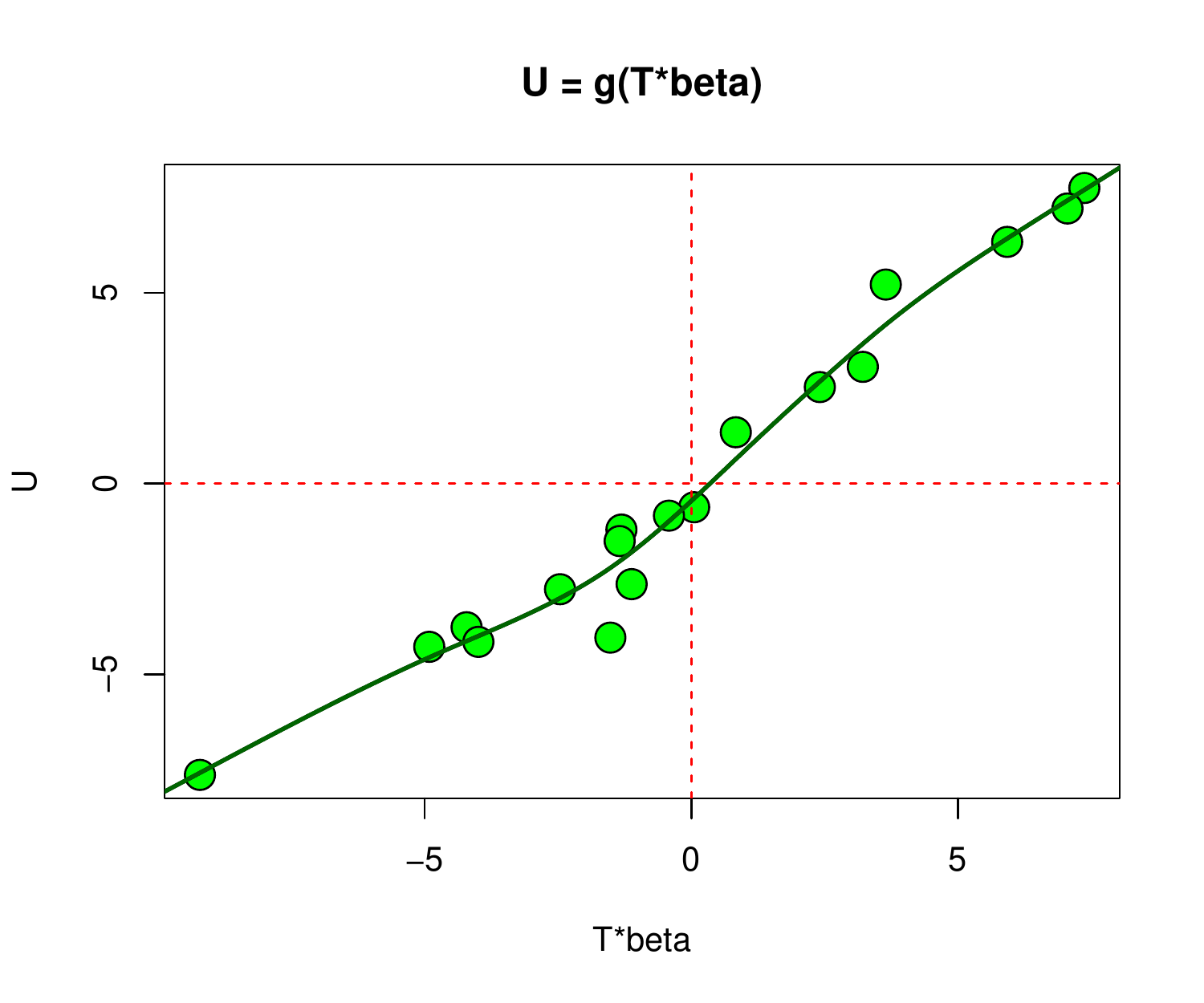}
	\caption{Left: screeplot using singular values of $X^TY$. Middle: Biplot for $X$. Right: Brillinger to find the relation between $T$ ($X$-scores) and $U$ ($Y$-scores). $\beta$ is a diagonally dominant matrix.}\label{diagnosis}
\end{figure}

Figure \ref{correlation} displays $X$ and $Y$ in a correlation circle (the outer circle has radius 1 and the inner one has radius 0.5). For each point, the coordinates in the plot are its correlations with the first two PLS components. This plot also indicates the correlation between variables. But unlike the biplot of $X$ in Figure \ref{diagnosis}, the variable locations are now determined by both $X$ and $Y$. We see that the large group of variables in the $X$-biplot are now separated into multiple subgroups. For example, sucrose, ph1, ph2, sweet, odor.typi are still close to each other, while bitter, acidity and citric form another subgroup.
\begin{figure}[H]
	\centering
	\includegraphics[width=0.5\linewidth]{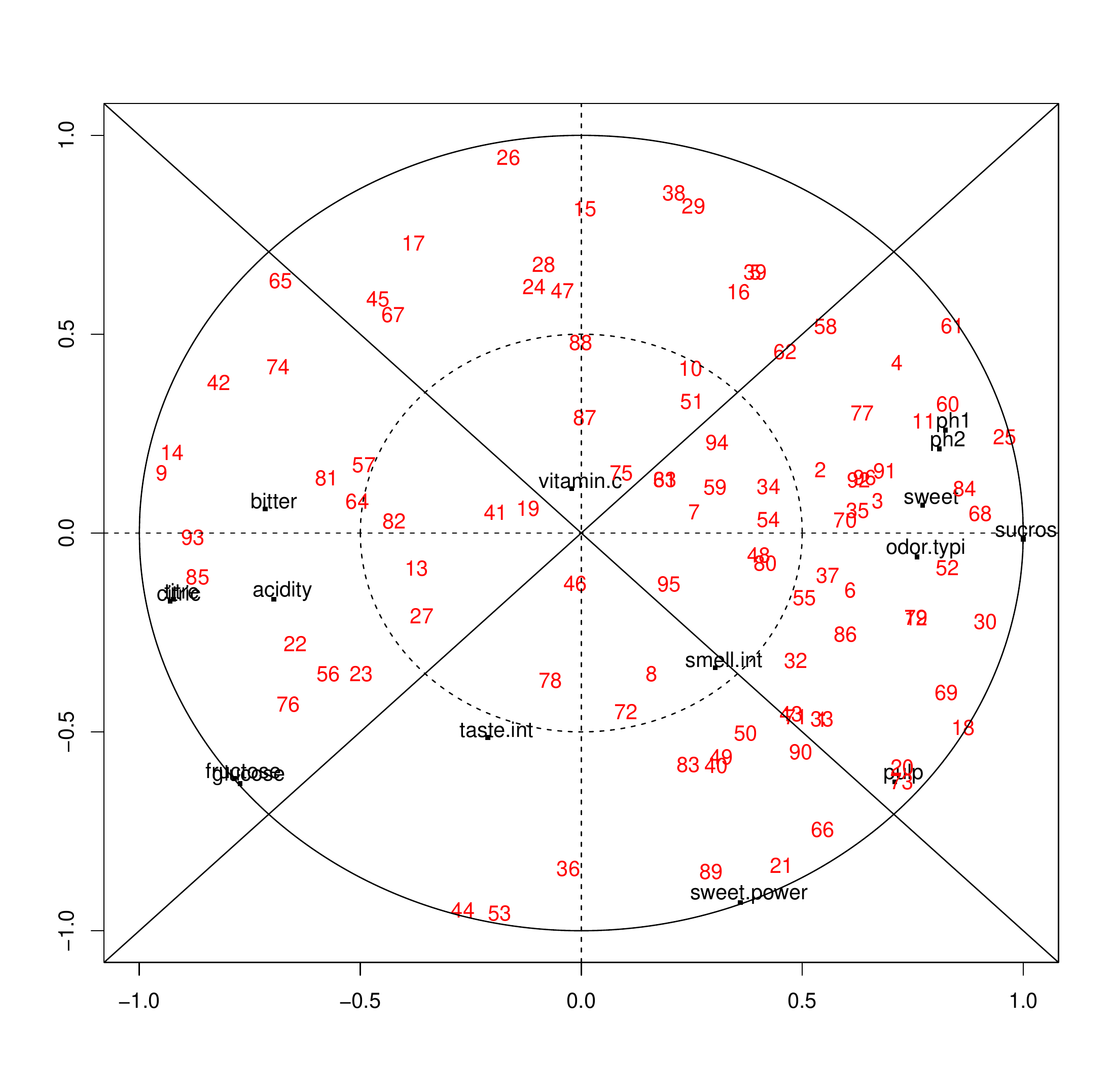}
	\caption{The correlation between $X$ ($Y$) and the first two PLS components $t_1, t_2$.}\label{correlation}
\end{figure}

Figure \ref{scaling} exhibits the scaling factor $f_j$ versus the number of components, where solid lines denote PLS and dotted ones are PCR. The factors for the first four judges are shown in the plot. For PLS, the number of components is 2. Hence the second eigen-direction is expanded and $f_j^{PLS} >1$ accordingly. The number of components for PCR is chosen so that the overall shrinkage are as close as possible. 
\begin{figure}[H]
	\centering
	\includegraphics[width=0.65\linewidth]{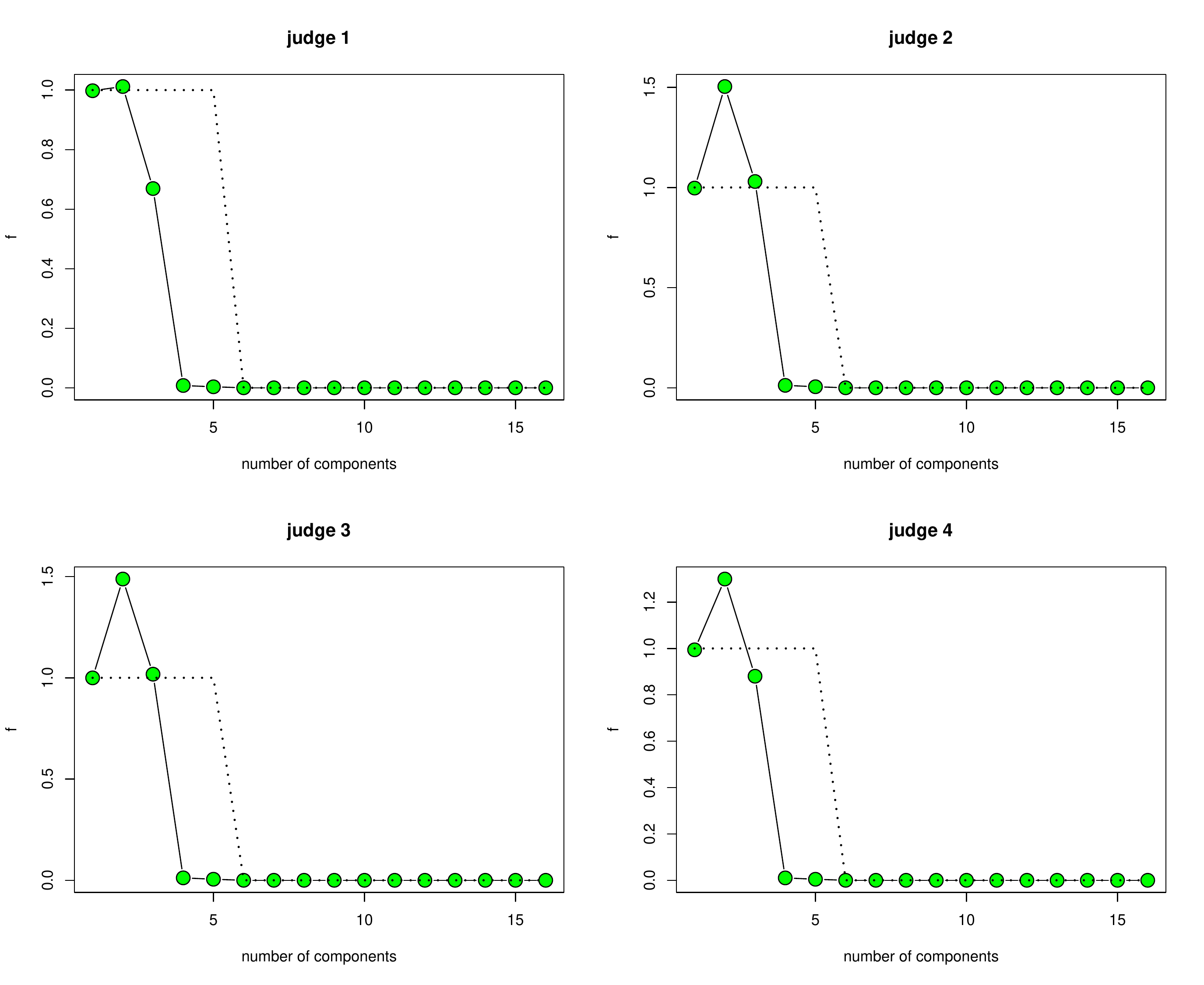}
	\caption{Scaling factor $f_j$ of PLS (solid lines) and PCR (dotted lines) for the first four judges.}\label{scaling}
\end{figure}

To predict $Y$ from its score vectors $U$ with a linear model, we assume independent and identical normal prior, $N(0, \sigma^2I_L)$ where $\sigma^2=0.1$, for coefficient vectors. Using the first five observations, we derive the corresponding posteriors as well as the posterior predictive distributions for the last observation. Results for the first three judges are shown in Figure \ref{predictive}. The parameters in the posterior predictive distribution are shown in the upper-right. Red lines denote the actual observed values.
\begin{figure}[H]
	\centering
	\includegraphics[width=0.9\linewidth]{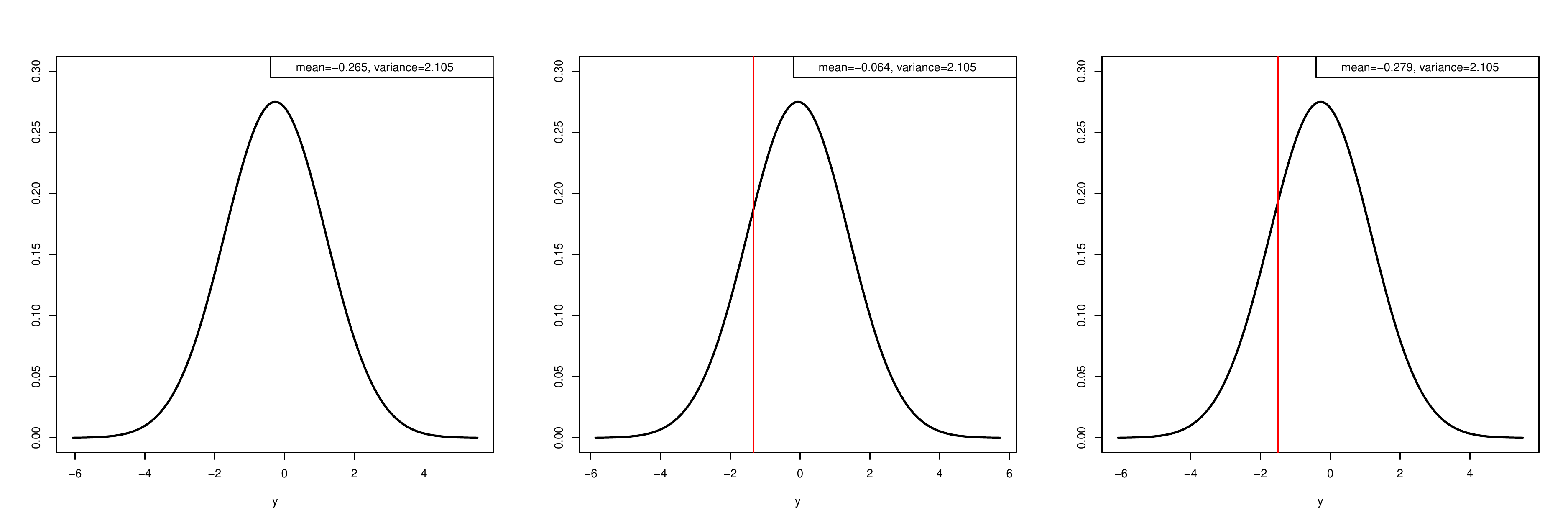}
	\caption{Posterior predictive distribution for the last observation of $Y$}\label{predictive}
\end{figure}

\subsection{Wine Data: Empirical distribution of the variables}

The Wine Quality Data Set (P. Cortez, A. Cerdeira, F. Almeida, T. Matos and J. Reis, `Wine Quality Data Set', UCI Machine Learning Repository.) contains 4\,898 observations with 11 features. The output  wine rating  is an integer variable ranging from 0 to 10 (the observed range in the data is from 3 to 9). 
 The frequency of each rating is as follows:
 \begin{table}[!ht]
 \caption{Frequencies of Wine Ratings}
 \centering
 \begin{tabular}{l | c c c c c c c }
 \toprule
 rating & 3 &   4 &   5 &   6 &   7 &   8 &   9  \\
 \midrule
 frequency & 20 & 163 & 1457 & 2198 & 880 & 175 &   5\\
 \bottomrule
 \end{tabular}
 \end{table}

Most of the wine falls into ratings 5 and 6. 
% To provide a binary classification problem, we bin the observations into two relatively balanced categories by choosing two types of comparison: (1) wine with ratings 5 and 6;  (2) wine with ratings $\leq 5$ and $>5$.

Figure \ref{histogram} shows the histograms of 11 characteristics after log transform or square root transform where appropriate in order for approximate Gaussianity to hold so that we can apply Brillinger's result. Figure \ref{quality} shows the scatterplots of wine quality versus wine characteristics.

First we use the trick of \cite{Hoadley2000} and expand the original 11 characteristic variables, denoted as $x_1, x_2, ..., x_{11}$, to a much higher dimensional set of 77 variables. These include 11 corresponding quadratics, i.e. $x_1^2, x_2^2, ..., x_{11}^2$, together with $11\times 10/2 = 55$ interactions, i.e. $x_1x_2,  ..., x_{10}x_{11}$. One of the advantages of PLS is that it can be applied in high-dimensional output and input variables that can be highly co-linear. Hence hand-coding all possible nonlinear functionals of characteristics and including them in the expanded set can be performed before using SVD to find the eigen-decomposition. 

\begin{figure}[H]
	\centering
	\includegraphics[width=0.8\linewidth]{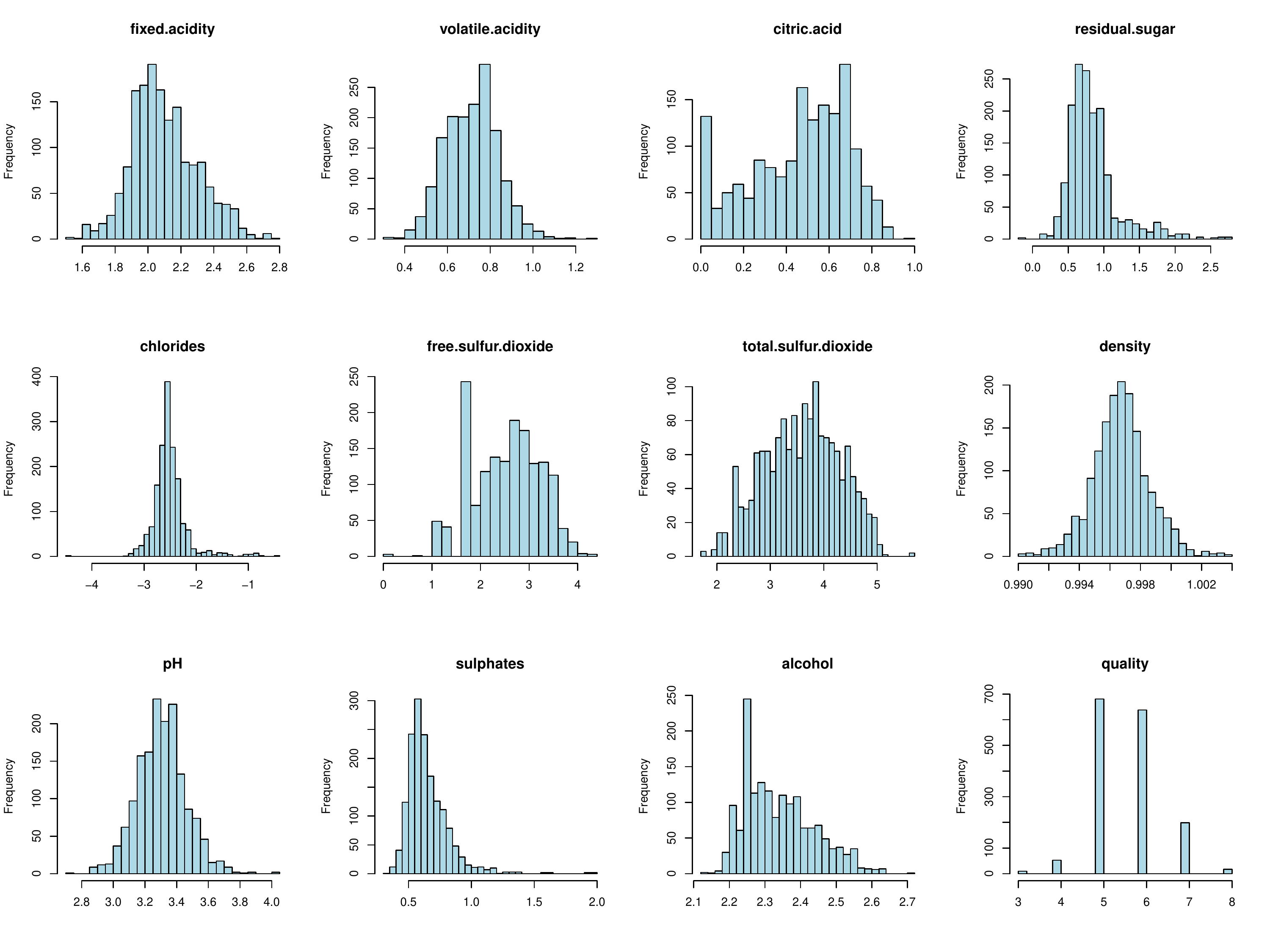}
	\caption{Histograms of wine variables after transformation}\label{histogram}
\end{figure}

\begin{figure}
	\centering
	\includegraphics[width=0.8\linewidth]{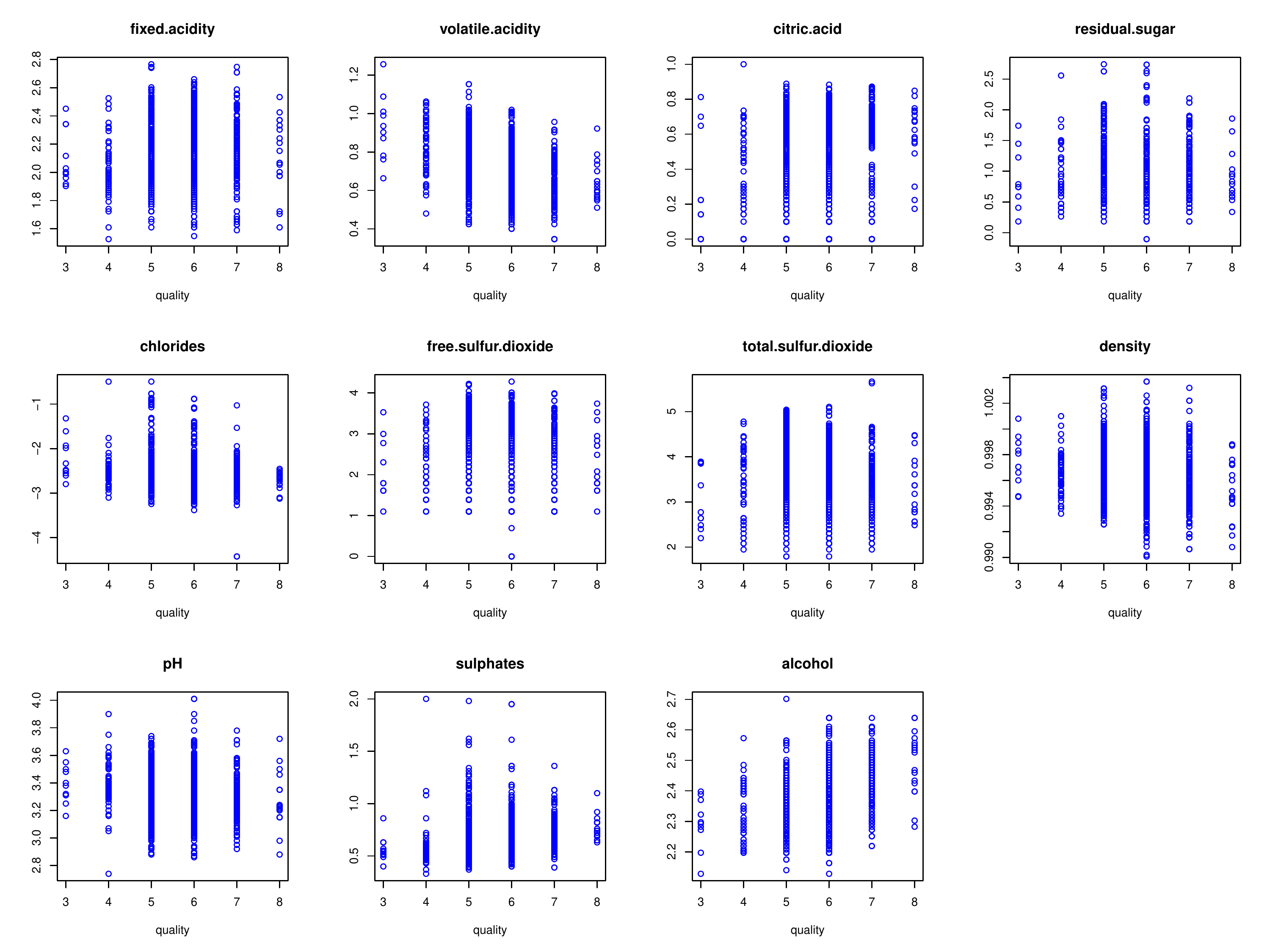}
	\caption{Wine variables versus wine quality}\label{quality}
\end{figure}

In Table \ref{compare}, we compare the in-sample adjusted $R^2$ and out-of-sample predictive accuracy of three linear models, OLS, PCR and PLS. The results for both input sets (the original 11 variables and the expanded 77 variables) are shown. For PCR, the number of components used are 5 and 15 for the original set and expanded set, respectively. For PLS, the number of components used are 5 and 10, respectively. As expected, the Hoadley trick improves the goodness of fit for all three models. It also helps with the out-of-sample classification for PCR and PLS. PLS has the best performance.

Based on the expanded variable set, we also build a neural network for quality classification (6 classes in total) and apply Brillinger's method. The architecture of our feed-forward neural network is 77-16-32-16-6. The second last layer with 16 hidden nodes is extracted as the new ``features" ($\tilde x$). We then run 6 regressions $y_i$ on $\tilde x$ to get $\hat\beta_i$'s, $i=1,2,...,6$. The nonlinear activation $g$ is estimated by smooth splines. The out-of-sample classification accuracy of the neural network is 0.889; after Brillinger's method being used, it's 0.885.

\begin{table}[ht]
	\centering
	\begin{tabular}{@{}lcc|cc@{}}
		\toprule    
		& \multicolumn{2}{c}{Reg. with 11 variables} & \multicolumn{2}{c}{Reg. with 77 variables} \\ \midrule
		& adj. $R^2$           & OOS acc.          & adj. $R^2$           & OOS acc.         \\ \midrule
		OLS                    & 0.347             & 0.59             & 0.403             & 0.575            \\
		PCR  & 0.324             & 0.58             & 0.372             & 0.595            \\
		PLS  & 0.347             & 0.595            & 0.389             & 0.6        \\ \bottomrule     
	\end{tabular}
	\caption{Comparison results for linear models}
	\label{compare}
\end{table}

\section{Discussion}\label{sec:discussion}
By merging deep learning with partial least squares, we provide a simple methodology for non-parametric input-out pattern matching. Partial least squares is a powerful statistical tool that performs data reduction and feature selection in high dimensional regression prediction. We introduce PLS-ReLU, PLS-trees and PLS-GP. Coupled with Brillinger's estimation insight from Stein's lemma and Gaussian regressors, it provides a data pre-processing tool with wide application in statistics and machine learning. 

Our procedure extends PLS by modeling the $Y$-scores ($U$) as a deep learner ($G_L$) of the input scores ($T$). This framework also provides a number of diagnostic plots to aid in building deep learning predictors. For example, the scree plot can be used to identify the number of factors. We illustrate Brillinger's result by showing how to recover the ReLU and tanh activation functions in a simulation study and we analyze two classic datasets, the orange juice preference hedonic regression (\cite{tenenhaus_pls_2005a}) and the wine quality dataset where we use the variable expansion approach of \cite{breiman2001statistical} and \cite{Hoadley2000}. 

Partial least squares can be used in many well-known machine learning models, such as kernel regression (\cite{rosipal2001kernel}, \cite{rosipal_kernel_2001}), Gaussian processes (\cite{gramacy2012gaussian}), tree models.  \cite{higdon2008computer} uses a PCA decomposition of the $Y$-variable before applying a nonlinear regression method. This can be addressed using our methodology. There are many directions for future research. For example, PLS for image processing, and extending Theorem 1 for larger classes of deep learners.

\bibliography{ref}

\end{document}